\documentclass[letterpaper,11pt]{article}
\usepackage{latexsym,amssymb,amsmath, bm}
\usepackage{amsthm}
\usepackage{lipsum,graphicx,multicol}
\usepackage[colorlinks=true,allcolors=blue, linkcolor=red]{hyperref}
\usepackage{setspace,color}
\usepackage[sort&compress]{natbib}
\usepackage{tcolorbox}
\usepackage{float}
\usepackage[normalem]{ulem}
\usepackage{authblk}
\usepackage{leftidx}
\usepackage{dsfont}
\usepackage{geometry}
\usepackage[title,titletoc]{appendix}
\usepackage{actuarialangle}
\allowdisplaybreaks

\usepackage{lineno}
\usepackage{bbm}

\usepackage{accents}

\def \Re {\mathbb{R}}

\def \uu {\bm{u}}

\def \XX {\bm{X}}
\def \xx {\bm{x}}

\def \E {\mathbb{E}}

\def \pp {\boldsymbol{p}}

\def \uu {\boldsymbol{u}}

\def \xx {\boldsymbol{x}}

\def \alpp {\boldsymbol{\alpha}}

\def \RR {\boldsymbol{R}}
\def \SS {\boldsymbol{S}}

\def \XX {\boldsymbol{X}}
\def \YY {\boldsymbol{Y}}

\def\VaR{{\mathrm{VaR}}}
\def\CTE{{\mathrm{CTE}}}
\def\e{{\mathbb E}}

\def\Cov{{\mathrm{Cov}}}

\newtheorem{theorem}{\bf Theorem}
\newtheorem{definition}{\bf Definition}
\newtheorem{question}{\bf Question}
\newtheorem{example}{\bf Example}

\newtheorem{proposition}{\bf Proposition}
\newtheorem{lemma}{\bf Lemma}
\newtheorem{corollary}{\bf Corollary}


\definecolor{darkread}{rgb}{0.7, 0, 0}

\begin{document}
\doublespace
\title{\Large\textbf{{Can a regulatory risk measure induce profit-maximizing risk capital allocations? The case of Conditional Tail Expectation}}}

\author{Nawaf Mohammed
\thanks{Department of Mathematics and Statistics, York University, Toronto, ON M3J 1P3, Canada.}\ \  \footnote{Corresponding author; postal address: 4700 Keele St, Toronto, ON M3J 1P3, Canada; email: {\href{mailto:nawaf@yorku.ca }{\color{magenta} nawaf@yorku.ca}} }\qquad
Edward Furman \thanks{Department of Mathematics and Statistics, York University, Toronto, ON M3J 1P3, Canada.}\qquad
Jianxi Su\thanks{Department of Statistics, Purdue University, West Lafayette, IN 47906, U.S.A.}
}

\date{}
\maketitle
\vspace{-1cm}
\begin{abstract}
Risk capital allocations (RCAs) are an important tool in quantitative risk management, where they are utilized to, e.g., gauge the profitability of distinct business units, determine the price of a new product, and conduct the marginal economic capital analysis. Nevertheless, the notion of RCA has been living in the shadow of another, closely related notion, of risk measure (RM) in the sense that the latter notion often shapes the fashion in which the former notion is implemented. In fact, as the majority of  the RCAs known nowadays are induced by RMs, the popularity of the two are apparently very much correlated. As a result, it is the RCA that is induced by the Conditional Tail Expectation (CTE) RM that has arguably prevailed in scholarly literature and applications. Admittedly, the CTE RM is a sound mathematical object and an important regulatory RM, but its appropriateness is controversial in, e.g., profitability analysis and pricing. In this paper, we address the question as to whether or not the RCA induced by the CTE RM may concur with alternatives that arise from the context of profit maximization. More specifically, we provide exhaustive description of all those probabilistic model settings, in which the \textit{mathematical} and \textit{regulatory} CTE RM may also reflect the risk perception of a profit-maximizing insurer.

\vspace{5mm}


\end{abstract}

\newpage
\section{Introduction}
\label{sec:introd}

Consider positive random variables (RVs) $X_1,\ldots,X_n,\ n\in\mathbb{N}$, which represent losses due to distinct business units (BUs) of an insurer, and denote by the sets $\mathcal{N}=\{1,\ldots,n\}$ and $\mathcal{X}$ the collections of these BUs and losses, respectively.  Then, for the aggregate loss RV $S_X:=X_1+\cdots+X_n$, the map $A:\mathcal{X}\times\mathcal{X}\rightarrow [0,\ \infty)$, which assigns {non-negative} values to random pairs $(X,S)\in\mathcal{X}\times\mathcal{X}$, is called a risk capital (RC) allocation rule \citep[e.g.,][]{Denault2001,Furman2008b,Dhaene2012}. Additionally, if $A(X,X)=H(X)$, where the map $H:\mathcal{X}\rightarrow [0,\ \infty)$ is called a risk measure and assigns {non-negative} values to the random loss $X\in\mathcal{X}$, then the allocation rule $A$ is said to be {induced} by the risk measure $H$.

RC allocation rules have gained major importance in risk management and insurance applications in the context of price determination, profitability assessment, budgeting decision making, to name just a few \citep[][]{Venter2004,Guo2018}. Similarly, the academic significance of - also, interest in - the subject of RC allocations have been strong, as evidenced by the large and growing body of scholarly literature  \citep[e.g.,][for recent references in the \textit{Insurance: Mathematics and Economics} journal, alone]{Boonen2019, KimKim2019, ShushiYao2020, Furmanetal18, Furmanetal2020}.

Not surprisingly, therefore, numerous RC allocation rules have been proposed and studied, with the RC allocation rule induced by the conditional tail expectation (CTE) risk measure being arguably the most popular \citep[][]{Kalkbrener2005}. More specifically, for $q\in [0,\ 1)$, $s_q:= \VaR_q(S_X)= \inf \big\{s\in[0,\ \infty):\ \mathbb{P}(S_X\leq s) \geq q\big\}$, loss portfolio $\XX=(X_1,\ldots,X_n)\in\mathcal{X}^n$ and BU $i\in\mathcal{N}$, the CTE-based RC allocation rule, when well-defined and finite, is given by
\begin{equation}
\label{CTE-alloc-def}
\CTE_q(X_i, S_X)=\e\left[
X_i|\, S_X>s_q
\right].
\end{equation}
As it is common in real applications to use RC allocation \eqref{CTE-alloc-def} - also, other RC allocation rules - to attribute the exogenous aggregate risk capital, say $\kappa\in\mathbb{R}_+$, to distinct BUs in the set $\mathcal{N}$, and, in order to guarantee the total additivity of the RC allocation rule, it is beneficial to explore \citep[e.g.,][]{Dhaene2012} the quantity,
$
\kappa_{i}=\kappa\times r_{q,i},
$
where
\begin{align}
\label{Simplexn-def}
r_{q,i}=\frac{\CTE_q(X_i, S_X)}{\CTE_q(S_X)}, \quad i\in\mathcal{N}
\end{align}
is the associated \textit{proportional} RC allocation rule induced by the CTE risk measure $\CTE_q(X)=\CTE_q(X,X)$ for any $X\in \mathcal{X}$ and $q\in [0,1)$.


The RC allocation rule based on the CTE risk measure has been thoroughly studied on a variety of fronts. Namely, allocation rule \eqref{CTE-alloc-def} was obtained as the gradient and the Aumann-Shapely allocation induced by the CTE risk measure
 by \cite{Tasche2004} and \cite{Denault2001}, respectively. Also, for loss RVs with continuous cumulative distribution functions (CDFs), allocation rule \eqref{CTE-alloc-def} coincides with the RC allocation rule induced by the Expected Shortfall risk measure \citep[][]{Kalkbrener2005,Wang.2019}. Last but not least, allocation rule \eqref{CTE-alloc-def} belongs to the class of distorted \citep[][]{TsanakasBarnett2003} and weighted \citep[][]{Furman2008b} RC allocation rules and is optimal in the sense of \cite{Laeven2004} and \cite{Dhaene2012}.

The number of works that evaluate allocation rule \eqref{CTE-alloc-def} for random losses having distinct joint CDFs is really overwhelming. For just a few examples, we refer to: \cite{Panjer2001} and \cite{LandsmanValdez2003} for, respectively, normal and elliptical distributions;  \cite{CaiLi2005} for phase-type distributions; \cite{Furman2005c} for gamma distributions; \cite{Vernic2006} and \cite{Vernic2011,Hendriks2017} for, respectively, skew normal and Pareto distributions; \cite{Furman2010} for Tweedie distributions; \cite{Cossette2012} for compound distributions with positive severities; \cite{Cossette2013,Ratovomirija2017} for mixed Erlang distributions; \cite{Furmanetal18} for Generalized Gamma Convolutions; this list is by no means exhaustive.

Despite the abundant relevant academic literature, RC allocation rule \eqref{CTE-alloc-def} - as well as the CTE risk measure that induces it - have been employed in practice mainly due to the inclusion in existing regulatory accords. When this aspect is put aside, the quantity $r_{q,i}$ raises a number of concerns. {First,  it hinges on a contentious two-step procedure, as the numerator and the denominator in Equation \eqref{Simplexn-def} must be computed separately for each BU $i\in\mathcal{N}$; similar concerns have been brought in \cite{Chong2019} in the context of the quantity $\kappa_i$. Second, it neglects the risk
perception of the insurer and the economic environment in which they operate \citep[][]{Bauer2016}.}

Admittedly, it is not surprising in any way that regulations, which are driven by the notion of prudence, and insurers' targets, which are profit-oriented, diverge. Nevertheless, it is instrumental to determine whether or not there exist model settings under which the RC allocation rule induced by the CTE risk measure yields outcomes that address the just-mentioned two concerns. This is what we do in the present paper.

We have organized the rest of this papers as follows. In Section \ref{Sec-2}, we motivate in detail and formulate the problem of interest. We then solve this problem in Sections \ref{sec:characteristics} and \ref{sec-ind}, which provide ample elucidating examples. Some of our analysis and conclusions carry over to a family of risk measures that contains the CTE risk measure as a special case, which is demonstrated in Section \ref{sec-afterthoughts}. Section \ref{Sec-conclusions} concludes the paper. In the sequel, we routinely work with an atomless and rich probability space $(\Omega,\mathcal{F},\mathbb{P})$, and we let $L^\alpha$ and $L^\infty$ denote, respectively, the set of all RVs that have finite $\alpha$-th moment, $\alpha\in[0,\ \infty)$, and the set of all essentially bounded RVs on the probability space $(\Omega,\mathcal{F},\mathbb{P})$. 
{Unless specified otherwise, we work with the collection of integrable RVs, $L^1$, so that the CTE risk measure and its associated RC allocation are well-defined and finite.}
Finally, we denote by $F_X$ and $\phi_X$ the CDF and the Laplace transform of the RV $X$, and we use $\mathbbm{1}$ to denote the indicator function.

\section{Compositional allocation rules induced by the Conditional Tail Expectation risk measure and a question that arises}
\label{Sec-2}

Note that the CTE-based allocation exercise \eqref{Simplexn-def} can be framed within the context of the standard $n$-dimensional simplex space \citep[][]{Aitchison1986}:
\begin{align*}
\mathfrak{S}^n=\big\{(r_1,\ldots,r_n):\, r_i\in [0,\ 1], \ i=1,\ldots,n\ \text{and}\ r_{1}+\cdots+r_{n}=1 \big\}.
\end{align*}
Specifically, for $x_1,\ldots,x_n\in [0,\ \infty)$ and $s:=x_1+\cdots+x_n$, as well as for the special map $\mathcal{C}:[0,\ \infty)^n\rightarrow \mathfrak{S}^n$ with $\mathcal{C}_i(x_1,\ldots,x_n)=x_i/s,\ i=1,\ldots,n$, proportional allocation rule \eqref{Simplexn-def} is obtained via setting $x_i=\CTE_q(X_i,S_X)$, and so \citep[][]{Belles-Sampera2016,Boonen2019}
\[
r_{q,i}=\mathcal{C}_i\big(\CTE_q(X_1,S_X),\ldots,\CTE_q(X_n,S_X)\big)=\CTE_q(X_i,S_X)\,/\,\CTE_q(S_X),\ q\in[0,\ 1).
\]
We note in passing that a similar reformulation of the RC allocation exercise in the context of the $n$-dimensional simplex space can be achieved effortlessly for the whole class of weighted RC allocation rules \citep[][]{Furman2008b}, which are induced by the class of weighted risk measures \citep[][]{Furman2008a} and of which allocation rule \eqref{CTE-alloc-def} is a special case \citep[][]{Furman2019a}.

An alternative way to determine the proportional contribution of the $i$-th BU of an insurer to the aggregate risk capital - under the assumption that it is the CTE risk measure that induces the desired allocation rule - is by considering the ratio RV $R_i=X_i\,/\,S_X,\ i=1,\ldots,n$, directly. That is, while, for a fixed $q\in[0,\, 1)$, the proportional allocation $r_{q,i}$, confined with the help of the normalizing constant $\CTE_q(S_X)\in\mathbb{R}_+$ to the unit interval, $\mathcal{I}=[0,\, 1]$, operates on random pairs $(X_i,S_X)\in\mathcal{X}\times\mathcal{X}$, an alternative to $r_{q,i}$ proportional allocation, call it $\tilde{r}_{q,i}$, is chosen to operate on random pairs $(R_i,S_X)\in\mathcal{I}\times\mathcal{X},\ i=1,\ldots,n$, and so
\[
\tilde{r}_{q,i}=\CTE_q\left(\mathcal{C}_i(X_1,\ldots,X_n),S_X\right)=
\CTE_q\left(
R_i,S_X
\right),\ q\in[0,\ 1).
\]

While various properties of the proportional allocation $r_{q,i}$ have been well-studied, this is not so for its counterpart, $\tilde{r}_{q,i}$. Further, we report a number of important properties of the latter quantity. In this respect, our first proposition shows that the quantity $\tilde{r}_{q,i}$ agrees with the economic capital allocation rule proposed recently by \cite{Bauer2016}. Namely, while the motivation for the RC allocation rule $r_{q,i}$ is the central role that the CTE risk measure plays in today's (insurance) regulation, the proportional allocation $\tilde{r}_{q,i}$ turns out to be  a well-justified choice for a profit maximizing insurer with risk-averse counterparties in an incomplete market setting with frictional capital costs.

Consider the aggregate loss RV $S_X\in\mathcal{X}$ and the Geometric Tail Expectation (GTE) risk measure:
\begin{equation}
\label{rtilde-rm-def}
{\rm GTE}_q(S_X):=\exp\big\{\mathbb{E}\left[\log(S_X)|S_X>s_q\right]\big\},\ q\in[0,\ 1).
\end{equation}
The connection of risk measure \eqref{rtilde-rm-def} to the notion of geometric means \citep[e.g.,][]{Hardyetal1952} motivates its name; also, risk measure \eqref{rtilde-rm-def} is a \textit{tail quasi-linear mean} risk measure in the sense of \cite{BaurleShushi2020}. It is not difficult to see that, for any $q\in[0,\ 1)$, risk measure \eqref{rtilde-rm-def} is at least as prudent as the Value-at-Risk risk measure and may be finite even if the CTE risk measure is infinite. Namely, we have the following simple result.
{
\begin{proposition}
For any $X\in\mathcal{X}$ and $q\in[0,\, 1)$, we have the bounds
\begin{align*}
    {\rm{VaR}}_{q}(S_X)\leq {\rm GTE}_q(S_X) \leq \CTE_q(S_X).
\end{align*}
\end{proposition}
\begin{proof}
By Jensen's inequality, we have, for $q\in[0,\, 1)$,
\[
\exp\big\{\mathbb{E}\left[\log(S_X)|S_X>s_q\right]\big\}\le \mathbb{E}\left[\exp\{\log(S_X)\}|S_X>s_q\right]=\mathbb{E}\left[S_X|S_X>s_q\right],
\]
which proves the upper bound. In addition, for $q\in[0,\, 1)$,
\[
s_q= \exp\big\{\mathbb{E}\left[\log(s_q)|S_X>s_q\right]\big\}\leq
\exp\big\{\mathbb{E}\left[\log(S_X)|S_X>s_q\right]\big\},
\]
establishing the lower bound and, hence, proving the proposition.
\end{proof}
}
Another immediate but worth-mentioning observation is that risk measure \eqref{rtilde-rm-def} is neither coherent in the sense of \cite{Artzner1999} nor convex in the sense of \cite{Follmer2001Convex}, as it violates translation-invariance. {Consequently, risk measure \eqref{rtilde-rm-def} is not a monetary risk measure. Nevertheless, risk measure \eqref{rtilde-rm-def} belongs to the class of return risk measures \citep[][]{Bellini2018}.} {{Moreover, when viewed through the prism of a profit maximizing insurer, risk measure \eqref{rtilde-rm-def} induces the optimal RC  allocation outcome, $\tilde{r}_{q,i}$, {in the sense of \cite{Bauer2016}};}} intuitively, this might be due to the decreasing marginal effect of the increase in aggregate loss \citep[][]{Bauer2016}. Our next statement about the RC allocation induced by risk measure \eqref{rtilde-rm-def} is formulated as a proposition.

{
Before stating our next result, we note that, similarly to how the CTE-based risk capital allocation is induced by the CTE risk measure, the GTE-based risk capital allocation,
\begin{align*}
{\rm GTE}_q(X_i,S_X):={\rm GTE}_q(S_X)\times \tilde{r}_{q,i},\quad q\in[0,\ 1),\ i\in\mathcal{N},
\end{align*}
is induced by risk measure \eqref{rtilde-rm-def}. The GTE-based RC is fully-additive, as $\sum_{i=1}^n {\rm GTE}_{q}(X_i, S)=
\sum_{i=1}^n {\rm GTE}_{q}\big(S_X\big) \times\tilde{r}_{q,i}={\rm GTE}_{q}\big(S_X\big)\times \sum_{i=1}^n \tilde{r}_{q,i}={\rm GTE}_{q}\big(S_X\big)$ for $q\in[0,\, 1)$.
\begin{proposition}\label{Gradient-prop}
The GTE-based RC allocation is the gradient allocation in the direction of the loss RV $X_i\in\mathcal{X},\ i\in\mathcal{N}$ induced by risk measure \eqref{rtilde-rm-def}.
\end{proposition}}
\begin{proof}
Since the GTE risk measure is positively homogeneous, by Euler's theorem we have,
for $\uu=(u_1,\dots,u_n)\in \mathcal{I}^n$ and $S_X(\uu):=u_1X_1+\cdots+u_nX_n$, 
\begin{align*}
{\rm GTE}_{q}\big(S_X(\uu)\big)=\sum_{i=1}^n u_i\,\frac{\partial}{\partial{u_i}}{\rm GTE}_q\big(S_X(\uu)\big),\quad \mbox{with  $q\in[0,\ 1)$, $ i\in\mathcal{N}$}.
\end{align*}
Therefore, for $\mathbf{1}$ denoting the $n$-variate vector of ones, we obtain
\begin{align*}
\frac{\partial}{\partial{u_i}} {\rm GTE}_{q}\big(S_X(\uu)\big) \bigg\vert_{\uu=\mathbf{1}}
&=\left({\rm GTE}_{q}\big(S_X(\uu)\big)
\times \E\left[\frac{X_i}{S_X(\uu)}\bigg\vert\ S_X(\uu)>\textnormal{VaR}_q(S_X(\uu))\right] \right)
\bigg\vert_{\uu=\mathbf{1}}\\[2mm]
&= {\rm GTE}_{q}\big(S_X\big) \times \tilde{r}_{q,i},\quad  \mbox{$i\in\mathcal{N}$.}
\end{align*}
This completes the proof of the proposition.
\end{proof}

The next assertion demonstrates that the proportional allocation induced by the CTE risk measure is an approximation of the GTE allocation induced by risk measure \eqref{rtilde-rm-def}.
\begin{proposition}\label{Linearapp-prop}
The proportional allocation $r_{q,i}$ is a linear approximation of the proportional allocation $\tilde{r}_{q,i}$ for $i\in\mathcal{N}$.
\end{proposition}
\begin{proof}
Consider the function $g(x_i,s)=x_i/s$ for $x_i,s\in\mathbb{R}_+$, $i=1,\ldots,n$, and denote its partial derivatives by
\begin{align*}g_i(x_i,s)=\frac{\partial}{\partial x_i}\,g(x_i,s)\quad \textrm{and}\quad g_s(x_i,s)=\frac{\partial}{\partial s}\,g(x_i,s).
\end{align*}
Then the first-order Taylor expansion of $g$ around $(x_0,s_0)=\big(\CTE_q(X_i\,,\, S_X),\, \CTE_q(S_X)\big)$, yields
\begin{align*}
x_i\,/\,s =   g(x_0,s_0) &+ g_i(x_0,s_0)\,(x_i - x_0) +g_s(x_0,s_0)\, (s - s_0 )+R_1(x_i,s),
\end{align*}
where $R_1(x_i,s)$ is the reminder term for all $x_i,s\in\mathbb{R}_+,\ q\in[0,\ 1)$, $i=1,\ldots,n$. Consequently, we have
\begin{align*}
    \tilde{r}_{q,i}\approx r_{q,i} +& g_i(x_0,s_0)\,\mathbb{E}\big[(X_i - x_0)|\ S_X>s_q\big] +g_s(x_0,s_0)\, \mathbb{E}\big[(S_X - s_0 )|\ S_X>s_q\big]=r_{q,i},
\end{align*}
which establishes the desired approximation and thus completes the proof of the proposition.
\end{proof}
Finally, our last proposition - which can be considered a follow-up on Proposition \ref{Linearapp-prop} - delineates the difference between the allocations $r_{q,i}$ and $\tilde{r}_{q,i}$. The proof is an immediate consequence of the identity
\[
\Cov(R_i,S_X|\ S_X>s_q)=\mathbb{E}[X_i|\ S_X>s_q]-\mathbb{E}[R_i|\ S_X>s_q]\times \mathbb{E}[S_X|\ S_X>s_q],
\]
which holds for all $q\in[0,\ 1)$ and $(X_i,S_X)\in\mathcal{X}\times\mathcal{X},\ i=1,\ldots,n$.
\begin{proposition}\label{Covrel-prop}
Given that all the quantities below are well-defined and finite, we have
\[
r_{q,i} =\tilde{r}_{q,i}+\frac{\Cov(R_i,S_X|\ S_X>s_q)}{\CTE_q(S_X)},\ q\in[0,\ 1),\ i\in\mathcal{N}.
\]
\end{proposition}
Proposition \ref{Covrel-prop} implies that it is the sign of the covariance between the RVs $R_i$ and $S_X$, that determines the order of the allocations $r_{q,i}$ and $\tilde{r}_{q,i}$; note that, as $R_1+\cdots+R_n=1$ almost surely, we have $\sum_{i=1}^n \Cov(R_i,S_X|\ S_X>s_q)=0$ \citep[][for examples, albeit in a different context, of the importance of covariances in insurance and finance]{Furman2010b}. Also, and more importantly, Proposition \ref{Covrel-prop}  suggests that the proportional allocation rules induced by the CTE risk measure and those induced by risk measure \eqref{rtilde-rm-def} coincide if the aforementioned covariance is nil. This motivates the following question that engages us in the rest of this paper.

\begin{question}\label{question}
For $X_i\in\mathcal{X}$, $i\in\mathcal{N}$, can we characterize those portfolios of losses $\XX=(X_1,\ldots,X_n)\in\mathcal{X}^n$, for which the proportional allocations $r_{q,i}$ and $\tilde{r}_{q,i}$ agree for every $q\in [0,\ 1)$?
\end{question}

At the outset, let us note that if the answer to the question above were in the affirmative, then this would imply that the regulatory (e.g., Swiss Solvency Test) CTE risk measure induces an optimal RC allocation rule for a profit maximizing insurer; the richer the class of joint CDFs of the RV $\XX\in\mathcal{X}^n$ sought in Question \ref{question}, the more common the just-mentioned and apparently desirable agreement between the regulatory requirements and the risk perceptions of insurers.

Speaking formally, our goal is to characterize the following collection of loss RVs:
\begin{align}
\label{eq:DSI}
  \mathfrak{W}=\big\{\XX=(X_1,\ldots,X_n)\in\mathcal{X}^n\ :\, r_{q,i}= \widetilde{r}_{q,i}\ \mbox{ for all $q\in [0,1)$ and $i\in\mathcal{N}$}\big\}.
\end{align}
It is to be noted that the set $\mathfrak{W}$ can not be empty. To see a trivial case in which $\tilde{r}_{q,i}=r_{q,i}$ for every $q\in[0,\ 1)$ and $i\in\mathcal{N}$, let the loss RV $\XX=(X_1,\ldots,X_n)\in\mathcal{X}^n$ have identically distributed coordinates, $X_i\in\mathcal{X}_i$, and an exchangeable copula function, then we have $\e[R_i|\ S_X=s]=1/n$ resulting in $\Cov(R_i,S_X|\,S_X=s)\equiv 0$, and therefore $\Cov(R_i,S_X|\, S_X>s_q)=0$ for all $q\in[0,\ 1)$ and $i=1,\ldots,n$. Consequently, the set of all loss RVs $\XX=(X_1,\ldots,X_n)\in\mathcal{X}^n$, such that the proportional allocation rules $r_{q,i}$ and $\tilde{r}_{q,i}$ coincide has at least one portfolio of losses in it. 

We devote the following sections of this paper to studying what other random loss RVs - besides the trivial example above - are members of the set $\mathfrak{W}$.

\section{General considerations}
\label{sec:characteristics}

In this section, we devise the necessary and sufficient conditions for the equality, $\tilde{r}_{q,i}=r_{q,i}$, for all $q\in[0,\ 1)$ and $i\in\mathcal{N}$. For this, we need a few auxiliary notions first. That is, Definitions \ref{sb-uv-def} and \ref{sb-mv-def} below introduce the univariate size-biased transform and its multivariate extension \citep{Patil1976, Arratiaetal2019, Furmanetal2020}, both playing major roles in our analysis.
\begin{definition}\label{sb-uv-def}
Let $X\in L^\alpha$ be a positive loss RV, then the size-biased counterpart of order $\alpha\in\mathbb{R}_+$ of the loss RV $X$, call it $X^{[\alpha]}$, is defined via:
\begin{equation}
\mathbb{P}\left(
X^{[\alpha]}\in dx
\right)=\frac{x^\alpha}{\e[X^\alpha]}\mathbb{P}\left(X\in dx\right)\quad \textnormal{ for all }x\in\mathbb{R}_+.
\end{equation}
When $\alpha=1$, we simply write $X^\ast$ for the size-biased of order one variant of the RV $X\in L^1$. The RVs $X$ and $X^{[\alpha]}$ are independent by construction for all $\alpha\in\mathbb{R}_+$.
\end{definition}
\begin{definition}
\label{sb-mv-def}
Let $\XX=(X_1,\ldots,X_n)\in\mathcal{X}^n$ be a loss RV with positive univariate coordinates $X_i\in L^{\alpha_i},\ i=1,\ldots,n$, then the multivariate size-biased counterpart of order $\alpp =(\alpha_1,\ldots,\alpha_n)\in \mathbb{R}_+^n$ of the loss RV $\XX$, call it $\XX^{[\alpp]}$, is defined via
\begin{equation}
\mathbb{P}\left(
\XX^{[\alpp]}\in d\xx
\right)=\frac{x_1^{\alpha_1}\times\cdots\times x_n^{\alpha_n}}{\e[X_1^{\alpha_1}\times\cdots\times X_n^{\alpha_n}]}\mathbb{P}\left(\XX\in d\xx\right)\quad \textnormal{ for all }\xx=(x_1,\ldots,x_n)\in\mathbb{R}_+^n.
\end{equation}
{The RVs $\XX$ and $\XX^{[\alpp ]}$ are independent by construction \citep[e.g.][for a similar discussion]{Patil1976, Arratiaetal2019, Furmanetal2020}.}
\end{definition}

We next define the \textit{partial size-biased transform}, which is a useful special case of the one presented in Definition \ref{sb-mv-def} \citep[e.g.,][for a few recent results in which the partial size-biased transform plays a central role but is not explicitly defined]{Arratiaetal2019, Furmanetal2020}.
\begin{definition}\label{def:multi-SB}
Consider the size-biased RV $\XX^{[\alpp ]},\ \alpp =(\alpha_1,\ldots,\alpha_n)$ as per Definition \ref{sb-mv-def}. Then, in the special case where the $i$-th coordinate of the vector $\alpp $ is equal to $\alpha_i=\alpha\in\mathbb{R}_+,\ i=1,\ldots,n$, whereas all other coordinates of the vector $\alpp $ are equal to zero, we call the implied transform, the $i$-th \textit{partial} size-biased transform of order $\alpha$, and denote the corresponding RV by $\XX^{[(\alpha)_i]}$. Namely, we have
\begin{equation}
\label{multi-SB-eq}
\mathbb{P}\left(
\XX^{[(\alpha)_i]}\in d\xx
\right)=\frac{x^{\alpha}_i}{\e[X_i^{\alpha}]}\, \mathbb{P}\left(\XX\in d\xx\right)\quad \textnormal{ for all }\xx=(x_1,\ldots,x_n)\in\mathbb{R}_+^n.
\end{equation}
{The RVs $\XX$ and $\XX^{[(\alpha)_i]}$ are independent by construction \citep[e.g.][for a similar discussion]{Patil1976, Arratiaetal2019, Furmanetal2020}.} For the case $\alpha=1$ and $X_i\in L^1$, we simply write, $\XX^{(\ast)_i}$, for the partial size-biased counterpart of the RV $\XX$.
\end{definition}
The operation of size-biasing has an important interpretation in the context of actuarial science and, more generally, in quantitative risk management, where it is considered \textit{loading} for model/sample risk. Indeed, it is not difficult to see that the size-biased loss RVs $X^{[\alpha]}$ and $\XX^{[\alpp]}$ (also, $\XX^{(\alpha)_i}$) dominate stochastically the loss RVs $X$ and $\XX$, respectively.

Furthermore, the partial size-biased RV, $\XX^{[(\alpha)_i]}$, plays an important role for \textit{size-biasing} sums of RVs. Namely, let $S_X^\ast=(X_1+\cdots+X_n)^\ast$, then the distribution of the RV $S_X^\ast$ admits a finite-mixture representation \citep[e.g.,][]{Arratiaetal2019} in terms of the partial size-biased RVs. Indeed, let $\phi_{S_X^\ast}$ denote the Laplace transform of the RV $S_X^\ast$, then, for $p_i=\e[X_i]\,/\,\e[S_X]$, we have
\begin{equation}\label{LTofSast-eq}
\phi_{S_X^\ast}(t)=\sum_{i=1}^n p_i\times \phi_{S_X^{(\ast)_i}}(t),\quad  {\rm Re}(t)>0,
\end{equation}
where $S_X^{(\ast)_i}$ is the sum of the coordinates of the partially size-biased RV $\XX^{(\ast)_i},\ i=1,\ldots,n$.

The following lemma is a variation of Equation \eqref{LTofSast-eq} that we find useful in this paper.

\begin{lemma}\label{Partsb-lemma}
Consider the RV $\XX_+=(X_1,\ldots,X_n,Y_{n+1},\ldots,Y_{n+m})\in\mathcal{X}^{n+m},\ n,m\in\mathbb{N}$, and let $S_X=\sum_{i=1}^n X_i,$ $S_Y=\sum_{i=n+1}^{n+m} Y_i$, $S_+=S_X+S_Y$, and $S_{+}^{(\ast)_X}=S_X^\ast+S_Y$. Then the distribution of the RV $S^{(\ast)_X}_{+}$ admits a mixture representation in the sense that we have $S^{(\ast)_X}_{+}=_{d} S^{(\ast)_K}_{+}$, where the RV $K\in\{1,\dots,n\}$ is such that $\mathbb{P}(K=k)=\e[X_i]\,/\,\e[S_X],\ k=1,\ldots,n$
\end{lemma}
\begin{proof}
Let $\phi_{S^{(\ast)_X}_{+}}$ denote the Laplace transform of the RV $S^{(\ast)_X}_{+}$, then, with the help of Equation \eqref{multi-SB-eq}, we have, for $p_i=\e[X_i]\,/\,\e[S_X]$,
\begin{eqnarray*}
\phi_{S^{(\ast)_X}_{+}}(t)&=& \frac{\e\left[S_X\,  e^{-tS_{+}}\right]}{\E[S_X]}=\sum_{i=1}^n p_i \frac{\e\left[X_i\, e^{-tS_{+}}\right]}{\e[X_i]}=\sum_{i=1}^n p_i\times \phi_{S_+^{(\ast)_i}}(t),\quad {\rm Re}(t)>0,
\end{eqnarray*}
which establishes the desired result and thus completes the proof.
\end{proof}

The next assertion spells out the sufficient and necessary conditions for the loss portfolios $\XX=(X_1,\ldots,X_n)\in\mathcal{X}^n$ to belong to the set $\mathfrak{W}$, and hence it answers Question \ref{question}.  The non-technical interpretation of the assertion is that for loss portfolios in the set $\mathfrak{W}$ and under the paradigm of loading for model/sample risk, the choice of the \textit{load direction} as per Definition \ref{def:multi-SB}
does not impact the distribution of the loaded aggregate loss RV.
\begin{theorem}
\label{thm:main}
Consider the loss RV $\XX=(X_1,\ldots,X_n)\in\mathcal{X}^n$ and assume that $X_i\in L^1$, then we have the equality $r_{q,i}=\tilde{r}_{q,i}(=\e[X_i]/\e[S_X])$ for all $q\in [0,\ 1),\ i\in\mathcal{N}$, if and only if $S_X^{(\ast)_i}=_{d}S_X^{(\ast)_j}(=_dS_X^{\ast}),\ i\neq j\in\mathcal{N}$.
\end{theorem}
\begin{proof}
\label{proof:thm-main}
Assume that $r_{q,i}=\tilde{r}_{q,i}$ for all $q\in [0,\ 1)$ and $i=1,\ldots,n$. By Proposition \ref{Covrel-prop}, this is equivalent to requesting that $\textrm{Cov}(R_i,S_X|\, S_X>u)=0$ for all $u\geq 0$, or, in other words with the notation $G_i(s)=\e[R_i|\ S=s]-\e[R_i],\ i=1,\ldots,n$, that
\[
\mathbb{E}\big[S_X\,G_i(S_X)|S_X>u\big]=\mathbb{E}\big[G_i(S_X)|S_X>u\big]\, \mathbb{E}[S_X|S_X>u]\quad  \textnormal{ for all }u\geq 0,
\]
from which we must have
\[
G_i(u)=\e[G_i(S_X)|\ S_X>u]\quad  \textnormal{ for all } u\geq 0.
\]
Hence, $G_i(u)\equiv\textnormal{const}$, which alongside the fact that $\e[G_i(S_X)]=0$, implies $G_i(u)\equiv 0$. Furthermore, as we assumed that $r_{q,i}=\tilde{r}_{q,i}$ for all $q\in [0,\ 1)$, we have $r_{0,i}=\tilde{r}_{0,i}$, and so
\[
\e[R_i|\, S_X]=\frac{\e[X_i]}{\e[S_X]}
\]
or, equivalently,
\[
\e[X_i|\, S_X]=\frac{\e[X_i]}{\e[S_X]}S_X
\]
for $i=1,\ldots,n$. Finally, we have the following implication in terms of the Laplace transform of the loss RV $S_X^{(\ast)_i}$ and for $i=1,\ldots,n$,
$$
\phi_{S_X^{(\ast)_{i}}}(t)=\frac{\mathbb{E}[X_i\,e^{-tS_X}]}{ \mathbb{E}[X_i]}=\frac{\mathbb{E}\big[\mathbb{E}[X_i|S_X]\, e^{-t\,S_X}\big]}{ \mathbb{E}[X_i]}=\frac{\mathbb{E}[S_X\, e^{-tS_X}]}{ \mathbb{E}[S_X]}=\phi_{S_X^{\ast}}(t),\quad \textrm{Re}(t)>0.
$$
This implies $S_X^{(\ast)_i}=_d S_X^{(\ast)_j}$ for all $1\leq i\neq j\leq n$ and so completes the `only if' direction of the theorem.

In order to prove the `if' direction of the theorem, let us assume that the distributional equality $S_X^{(\ast)_i}=_d S_X^{(\ast)_j}(=_dS_X^{\ast})$ holds for all $i=1,\ldots,n$, which means
\begin{align*}
\frac{\mathbb{E}[X_i\,e^{-t{S_X}}]}{ \mathbb{E}[X_i]}=\frac{\mathbb{E}[S_X\,e^{-t{S_X}}]}{\mathbb{E}[S_X]}
\end{align*}
or, equivalently,
\begin{align*}{\mathbb{E}\left[\frac{\mathbb{E}[X_i\,|\,{S_X}]}{ \mathbb{E}[X_i]}\,e^{-t{S_X}}\right]}={\mathbb{E}\left[\frac{S_X}{ \mathbb{E}[S_X]}\,e^{-t{S_X}}\right]},
\end{align*}
with the immediate implication
\begin{align*}
{\mathbb{E}[R_i|{S_X=u}]}={\frac{\e[X_i]}{\mathbb{E}[S_X]}}\quad  \textnormal{ for all }u\geq 0.
\end{align*}
Therefore, we necessarily have ${\mathbb{E}[R_i]}=\e[X_i]\,/\,\e[S_X]$, $i=1,\ldots,n$. Finally, we obtain the following string of equations:
\begin{align*}
\mathrm{Cov}(R_i,{S_X}\,|\,{S_X}>u)&=\mathbb{E}[R_i\,{S_X}\,|\,{S_X}>u]-\mathbb{E}[R_i\,|\,{S_X}>u]\, \mathbb{E}[{S_X}\,|\,{S_X}>u]
\\
&={\mathbb{E}\big[\mathbb{E}[R_i|{S_X}]\, {S_X}\big|{{S_X}>u}\big]}-{\mathbb{E}\big[\mathbb{E}[R_i|{S_X}]\big|{{S_X}>u}\big]}\,\mathbb{E}[{S_X}|{S_X}>u]
\\
&=\frac{\mathbb{E}[X_i]}{\mathbb{E}[{S_X}]}\ \mathbb{E}[{S_X}|{S_X}>u]-\frac{\mathbb{E}[X_i]}{\mathbb{E}[{S_X}]}\ \mathbb{E}[{S_X}|{S_X}>u]
\\
&=0
\end{align*}
for all $u\geq 0$ and $i=1,\ldots,n$.
The `if' direction of the theorem is then proved by evoking Proposition \ref{Covrel-prop}. This completes the proof of the theorem.
\end{proof}
Some properties of the portfolios of losses $\XX\in\mathfrak{W}$ are studied next. Specifically, it turns out that these portfolios are \textit{consistent} in the sense that the answer to Question \ref{question} must be in affirmative for all their sub-portfolios. This is formulated and proved next.
\begin{theorem}
\label{prop:split}
Consider the loss RV $\XX_+=(X_1,\ldots,X_n,Y_{n+1},\ldots,Y_{n+m})\in\mathcal{X}^{n+m}$ and assume that $\XX_+\in\mathfrak{W}$, then the sub-portfolios $(X_1,\ldots,X_n)$ and $(Y_{n+1},\ldots,Y_{n+m})$ also belongs to the set $\mathfrak{W}$.
\end{theorem}

\begin{proof}
We prove that if $\mathbf{X}_+\in\mathfrak{W}$, then $\XX=(X_1,\ldots,X_n)\in \mathfrak{W}$; the case $\YY=(Y_{n+1},\ldots,Y_{n+m})\in \mathfrak{W}$ follows in the same fashion.
Let $S_X=\sum_{i=1}^n X_i$, $S_Y=\sum_{i=n+1}^{n+m} Y_i$ and $S_+=S_X+S_Y$, as in Lemma \ref{Partsb-lemma}.
Because $\XX_+\in\mathfrak{W}$ and by Theorem \ref{thm:main}, we have, for all $i\neq j\in\{1,\ldots,n\}$ and ${\rm Re}(t)>0$,
\begin{align}
\phi_{S_+^{(\ast)_i}}(t)  :=\e\big[\exp\{-t\, S_+^{(\ast)_i}\}\big] & = \frac{\e\big[X_i\,  e^{-t(S_X+S_Y)}\big]}{\e[X_i]}
  \nonumber \\[2mm]
  & =\frac{\e\big[X_j\, e^{-t(S_X+S_Y)}\big]}{\e[X_j]}=\e\big[\exp\{-t\, S_+^{(\ast)_j}\}\big]=:\phi_{S_+^{(\ast)_j}}(t).
\end{align}
Therefore, we have
\begin{align*}
\e\Bigg[e^{-t\, S_Y}\,  \e\left[
\frac{X_i}{\e[X_i]}e^{-tS_X}\Big|\, S_Y
\right]\Bigg]=\e\Bigg[e^{-t\,S_Y}\, \e\left[
\frac{X_j}{\e[X_j]}e^{-tS_X}\Big|\, S_Y
\right]\Bigg]\quad \mbox{ for all ${\rm Re}(t)>0$},
\end{align*}
from which we can conclude
\begin{align*}
\e\left[
\frac{X_i}{\e[X_i]}\,e^{-tS_X}\Big|\, S_Y
\right]=\e\left[
\frac{X_j}{\e[X_j]}\,e^{-tS_X}\Big|\, S_Y
\right]\quad \mbox{ for all ${\rm Re}(t)>0$}.
\end{align*}
The assertion follows by the law of total expectation and evoking again Theorem \ref{thm:main}.
\end{proof}

Theorem \ref{prop:split} remains true if a \textit{split} results in more than two loss portfolios and implies that, when starting with a loss portfolio in the set $\mathfrak{W}$, the {split} operation yields loss portfolios that are also in the set $\mathfrak{W}$.

The next result emphasizes that the \textit{merge} operation - an opposite of split - is more intricate, but that the {merge} of loss portfolios belonging to the set $\mathfrak{W}$ may result in a loss portfolio that also belongs to the set $\mathfrak{W}$.
\begin{theorem}
\label{merge-prop}
\allowbreak
Consider {{two}} independent loss portfolios, $(X_1,\ldots,X_n)\in\mathfrak{W}$ and $(Y_{n+1},\ldots,Y_{n+m})\in\mathfrak{W}$, and denote by $S_X$ and $S_Y$ the corresponding sums of coordinates. Also, let $\XX_+=(X_1,\ldots,X_n,Y_{n+1}, \allowbreak\ldots,Y_{n+m})\in\mathcal{X}^{n+m}$ be the merged portfolio. Then, $\XX_+\in\mathfrak{W}$ if and only if, for $i\in\{1,\ldots,n\}$ and $j\in\{n+1,\ldots,n+m \}$,
\begin{align}
\label{eq:merge}
\frac{\phi_{S_X^{(\ast)_i}}(t)}{\phi_{S_Y^{(\ast)_j}}(t)}=\frac{\phi_{S_X}(t)}{\phi_{S_Y}(t)},\quad \mbox{${\rm Re}(t)>0$}.
\end{align}
\end{theorem}
\begin{proof}
Let $S_+=S_X+S_Y$.  We need to show that $S_+^{(\ast)_i}=_d S_+^{(\ast)_j}$ for all $i\neq j\in\{1,\ldots,n+m\}$. First, consider the case in which the indices $i,j$ belong to either one of the sets $\{1,\ldots,n\}$ or  $\{n+1,\ldots,n+m\}$, say $i\in \{1,\ldots,n\}$ and $j\in \{n+1,\ldots,n+m\}$.  Then by Lemma \ref{Partsb-lemma} with the addition of the independence assumption and since $(X_1,\ldots,X_n)\in\mathfrak{W}$ and $(Y_{n+1},\ldots,Y_{n+m})\in\mathfrak{W}$, we have
\begin{align*}
\phi_{S_+^{(\ast)_i}}(t)&=\phi_{S_X^{(\ast)_i}}(t)\times \phi_{S_Y}(t) =\phi_{S_Y^{(\ast)_j}}(t)\times \phi_{S_X}(t)=\phi_{S_+^{(\ast)_j}}(t)
\end{align*}
for all ${\rm Re}(t)>0$ if and only if Equation \eqref{eq:merge} is valid.

The case when $i\neq j$ are both in $\{1,\ldots,n\}$ or both in $\{n+1,\ldots,n+m\}$ follows similarly. This completes the proof of the assertion.
\end{proof}


The assertion that concludes this section reveals that an \textit{amalgamation} - on a BU basis - of a collection of loss portfolios, each of which belongs to the set $\mathfrak{W}$, may result in a loss portfolio that also belongs to the set $\mathfrak{W}$.
\begin{theorem}
\label{aggregate-prop}
Consider {{two}} independent loss portfolios, $(X_1,\ldots,X_n)\in\mathfrak{W}$ and $(Y_{1},\ldots,Y_{n})\in\mathfrak{W}$. Let $\SS=(S_1,\ldots,S_n)$, where $S_i=X_i+Y_i$, $i=1,\ldots,n$.  Then $\SS\in\mathfrak{W}$ if and only if
\begin{align}
\label{eq-aggregate-prop}
\frac{\e[X_{i}]}{\e[X_{j}]}=\frac{\e[Y_{i}]}{\e[Y_{j}]}\quad \mbox{for all $i\neq j\in\mathcal{N}$}.
\end{align}
\end{theorem}
\begin{proof}
Let $S_X=X_1+\cdots+X_n$, $S_Y=Y_1+\cdots+Y_n$, and $S_+=S_X+S_Y$.  By Lemma \ref{Partsb-lemma} with the addition of the independence assumption, we have, for $i=1,\ldots,n$,
\begin{align*}
 \phi_{S_+^{(\ast)_{i}}}(t)=\frac{\e\big[S_i\, e^{-t\, S_+}\big]}{\e[S_i]}=\frac{\e[X_i]}{\e[S_i]}\,\phi_{S_X^{(\ast)_{i}}}(t) \,
 \phi_{S_Y^{(\ast)}}(t) + \frac{\e[Y_i]}{\e[S_i]}\,\phi_{S_Y^{(\ast)_{i}}}(t) \,
 \phi_{S_X^{(\ast)}}(t)
 \quad \mbox{ for all ${\rm Re}(t)>0$}.
\end{align*}
By Theorem \ref{thm:main}, Condition \eqref{eq-aggregate-prop} is required so that the equality $\phi_{S_+^{(\ast)_{i}}}(t)=\phi_{S_+^{(\ast)_{j}}}(t)$ holds for all ${\rm Re}(t)>0$ and $i\neq j\in \mathcal{N}$. This completes the proof of the assertion.
\end{proof}

{We conclude this sub-section with the note that Theorems \ref{merge-prop} and \ref{aggregate-prop} stay valid even if more than two loss portfolios are considered. Specifically, consider $m(\in \mathbb{N})$ loss portfolios $\XX_i=(X_{i,1},\ldots,X_{i,n_i})\in \mathcal{X}^{n_i}$ each have $n_i$ BUs, where $n_i\in \mathbb{N}$, $i=1,\ldots,m$.  Assume that each portfolio $\XX_i\in \mathfrak{W}$, and $(\XX_1,\ldots,\XX_m)$ are mutually independent.  Let $S_{\XX_i}=X_{i,1}+\ldots+X_{i,n_i}$, then $\XX_+=(\XX_1,\ldots,\XX_m)\in \mathfrak{W}$ if and only if
\begin{align*}
 {\phi_{S_{\XX_i}^{(\ast)_{k_i}}}(t)\over \phi_{S_{\XX_j}^{(\ast)_{k_j}}}(t)}=  {\phi_{S_{\XX_i}}(t)\over\phi_{S_{\XX_j}}(t)}\quad \mbox{for all $k_i\in \{1,\ldots,n_i\}$, $k_j\in \{1,\ldots,n_j\}$, $i\neq j\in\{1,\dots,m\},$}\ \mbox{${\rm Re}(t)>0$,}
\end{align*}
which is a multi-portfolio adjustment of Condition \eqref{eq:merge}.  Moreover, let $n_1=\cdots=n_m=n,$ $S_j=X_{1,j}+\cdots+X_{m,j}$, $j\in\{1,\ldots,n\}$, then $\SS=(S_1,\ldots,S_n)\in \mathfrak{W}$ if and only if
\begin{align*}
 \frac{\e[X_{1,i}]}{\e[X_{1,j}]}=\cdots=\frac{\e[X_{m,i}]}{\e[X_{m,j}]}\quad \mbox{for all $i\neq j\in \{1,\ldots,n\}$,}
\end{align*}
which is a multi-portfolio adjustment of Condition \eqref{eq-aggregate-prop}.
}


\subsection{Examples and further elaborations}
In this section, we review a few examples of those loss RVs, $\XX\in\mathcal{X}^n$, for which the equality $r_{q,i}=\tilde{r}_{q,i}$ holds for all $q\in[0,\ 1)$ and $i\in\mathcal{N}$. That is, we now construct a few examples of the loss portfolios $\XX\in\mathcal{X}^n$, for which the RC allocations induced by the CTE risk measure reflect the diminishing impact of large losses on the insurers' perception of risk.

Our first example is the Liouville distributions (e.g., \citeauthor{Gupta1987}, \citeyear{Gupta1987}, for a comprehensive treatment, and \citeauthor{Hua2016}, \citeyear{Hua2016}; \citeauthor{McNeil2010a}, \citeyear{McNeil2010a}, for applications in the context of dependence modelling).  To start with, for $\gamma\in\mathbb{R}_+$, denote by $\Gamma(\gamma)$ the complete gamma function, that is
\begin{align*}
\Gamma(\gamma)=\int_0^\infty x^{\gamma-1}e^{-x} dx.
\end{align*}
Also, for $\gamma_1,\ldots,\gamma_n\in\mathbb{R}_+$ and $\gamma_{\bullet}:=\gamma_1+\cdots+\gamma_n$, define the multivariate Beta function as
\begin{align*}
B(\gamma_1,\ldots,\gamma_n)=\frac{\prod_{j=1}^n \Gamma(\gamma_j)}{\Gamma(\gamma_{\bullet})}.
\end{align*}
\begin{example}[]
\label{Liouville-ex}
The positive and absolutely-continuous RV, $\XX=(X_1,\ldots,X_n)\in\mathcal{X}^n$, is said to be distributed Inverted-Dirichlet, succinctly $\XX\sim {\rm ID}_n(\gamma_1,\ldots,\gamma_n,\beta)$ with the parameters $\beta,\gamma_1,\ldots,\gamma_n\in\mathbb{R}_+$, if its probability density function (PDF) is:
\begin{align*}
\label{Liouville-pdf-eq}
f_{\XX}(x_1,\ldots,x_n)=\frac{1}{B(\gamma_1,\ldots,\gamma_n,\beta)}\,\prod_{j=1}^n x_j^{\gamma_j-1}\left(1+\sum_{j=1}^n x_j\right)^{-(\gamma_{\bullet}+\beta)},\quad x_1,\ldots,x_n\in\mathbb{R}_+,
\end{align*}
\citep[e.g.,][for a general discussion and applications in actuarial science, respectively]{Gupta.1996,Ignatov2004}.

It is not difficult to show that $\phi_{S^{(\ast)_i}}(t)=\phi_{S^{(\ast)_j}}(t),\ {\rm Re}(t)>0$ for all $1\leq i\neq j\leq n$, and hence by Theorem \ref{thm:main}, we have $r_{q,i}={\gamma_i}/{
\gamma_{\bullet}}=\tilde{r}_{q,i}$
for $q\in [0,\ 1)$ and $i\in\mathcal{N}$.
\end{example}

An interesting observation that paves the way for a fairly general proposition, which is stated next, is that for $\XX\sim {\rm ID}_n(\gamma_1,\ldots,\gamma_n,\beta)$, we have the stochastic representation $X_j=Z\times Y_j,\ j=1,\ldots,n$, where the RV $Z=\sum_{j=1}^n X_j$ has a univariate inverted beta distribution, $Z\sim IB(\gamma_\bullet,\beta)$, with the parameters $\gamma_\bullet,\beta\in\mathbb{R}_+$, and the RV $\YY=(Y_1,\ldots,Y_n)$, independent on the RV $Z$, is distributed multivariate Dirichlet \citep{Ng2011}.

The proof of the following assertion is readily obtained via the routine conditioning and then evoking Theorem \ref{thm:main} and is thus omitted.
\begin{proposition}
\label{prop:lous}
Let the RV $\YY=(Y_1,\ldots,Y_n)$ be independent on the RV $Z$ and such that, for a constant $b\in\Re_+$, the equality, $\sum_{i=1}^n Y_i= b$, holds almost surely. Further, let the loss portfolio  $\XX=(X_1,\ldots,X_n)\in\mathcal{X}^n$ admit the stochastic representation $X_j=Y_j\times Z,\ j\in\mathcal{N}$, then $\XX\in\mathfrak{W}$.
\end{proposition}

Proposition \ref{prop:lous} implies that the loss RVs $X_1,\ldots,X_n$ that admit the Multiplicative Background Risk Model (MBRM) stochastic representation with the idiosyncratic risk factors (RFs), $Y_1,\ldots,Y_n$  distributed Dirichlet with parameters $\gamma_1,\ldots,\gamma_n\in\mathbb{R}_+$ and the systemic RF $Z$ having the PDF $ f_Z$, such that
\begin{equation}
\label{eq:PDF_Z}
    f_Z(z) \propto g(z)\, z^{\gamma_\bullet-1},\ z\in\mathbb{R}_+,
\end{equation}
where $\gamma_{\bullet}=\sum_{i=1}^n\gamma_i$, the function $z\mapsto g(z)$ is positive, continuous and integrable in the sense of \cite{Gupta1987}, all belong to the set $\mathfrak{W}$. Some examples, in addition to the already-mentioned inverted beta distribution, of the probability distribution of the systemic RF, $Z$, are: the gamma distribution and the generalized mixture of exponential distributions.

The class of multivariate probability distributions that admit the stochastic representation described in Proposition  \ref{prop:lous} is called the class of \textit{Liouville distributions}, and these distributions are one way to extend the multivariate Dirichlet distribution to the unbounded domain, $\mathbb{R}_+^n$. Another way is via the class of mixed-Gamma  (MG) distributions, which has recently been presented and studied in \cite{Furman2019a}. Speaking briefly and avoiding unnecessary technicalities - thus considering the simplest possible case - the loss RV $\XX=(X_1,\ldots,X_n)$ is said to be distributed $n$-variate MG distribution if it has the PDF:
\begin{align}
\label{eq:MMG-PDF-diff}
    f_{\XX}(x_1,\ldots,x_n)= \sum_{k=1}^m p_k\, \prod_{i=1}^n \frac{x_i^{\gamma_{i,k}-1}}{\Gamma(\gamma_{i,k})\, \theta_i^{\gamma_{i,k}}}\, e^{-x_i/\theta_i},\quad x_1,\ldots,x_n\in\Re_+,
\end{align}
where $\gamma_{i,k}\in\mathbb{R}_+$ and $\theta_i\in\mathbb{R}_+$ are, respectively, the shape and scale parameters, and $p_k>0,\ k=1,\ldots,m$ are the mixture weights satisfying $\sum_{k=1}^m p_k=1$;
succinctly, we write $\XX\sim {\rm MG}_n(\boldsymbol{\gamma},\boldsymbol{\theta},\pp)$, where  $\boldsymbol{\gamma}$ and $\boldsymbol{\theta}$ are the $n\times m$- and $m$- dimensional vectors of shape and scale parameters, respectively, and $\pp=(p_1,\ldots,p_m)$.

The class of MG distributions is a generalization of the popular class of multivariate Erlang mixtures considered in \cite{Willmot2014}, albeit with (a) positive - and not positive and integer - shape parameters, and (b) possibly distinct -  and not all equal - scale parameters \citep[e.g.][]{Verbelen2016, Lee2012}. The class of MG distributions is connected to Question \ref{question} in the example and proposition that follow.

\begin{example}
\label{ex:MMG-same-scale}
Consider a loss portfolio $\XX\sim {\rm MG}_n(\boldsymbol{\gamma},\boldsymbol{\theta},\pp)$  with the PDF as per Equation (\ref{eq:MMG-PDF-diff}), but with $\theta_i\equiv \theta$. Then, for $i\in \mathcal{N}$, we have
\begin{align*}
\phi_{S_X^{(\ast)_i}}(t)=\e[\exp\{-t S_X^{(\ast)_i}\}]=\sum_{k=1}^m p^{(*)_i}_{k}\left(1+\theta t\right)^{-\gamma_{\bullet,k}-1},\quad {\rm Re}(t)>0,
\end{align*}
where
\begin{align*}
p^{(*)_i}_{k}=\frac{\gamma_{i,k}}{\sum_{k=1}^m \gamma_{i,k}\times p_k}\, p_k,\quad k=1,\ldots,m,
\end{align*}
which can be viewed as the  $i$-th partial
size-biased transform of the PMF underlying the stochastic shape parameters.
Consequently, for the equality $S_X^{(\ast)_i}=_d S_X^{(\ast)_j}$ to hold, we must require (due to Theorem \ref{thm:main})
\begin{align}\label{p-new-eq}
\frac{\gamma_{i,1}}{\left(\sum_{j=1}^n\gamma_{j,1}\right)}=\cdots=\frac{\gamma_{i,m}}{\left(\sum_{j=1}^n\gamma_{j,m}\right)}
\end{align}
for all $i\in\mathcal{N}$.
\end{example}

The observation presented in Example \ref{ex:MMG-same-scale} is strengthened in the following proposition, which concludes this section.
\begin{proposition}
\label{prop:MMG-single-theta}
Let $\XX\sim {\rm MG}_n(\boldsymbol{\gamma},\boldsymbol{\theta},\pp)$, then we have
 $\XX\in \mathfrak{W}$ if and only if both of $\theta_i\equiv \theta$ and Equation \eqref{p-new-eq} hold true.
\end{proposition}

\begin{proof}
Example \ref{ex:MMG-same-scale} establishes the `if' direction. In order to prove the `only if' direction,  we  pursue proof by contradiction. To this end, consider $\XX\sim {\rm MG}_n(\boldsymbol{\gamma},\boldsymbol{\theta},\pp)$ in which the coordinates of the vector of parameters $\boldsymbol{\theta}=(\theta_1,\ldots,\theta_n)$ are all distinct, and suppose $\XX\in \mathfrak{W}$. (If some scale parameters were equal, then we would introduce the vector of distinct scales, $\boldsymbol{\widehat{\theta}}=(\widehat{\theta}_1,\ldots,\widehat{\theta}_{n'})$, $n'< n$ as well as, for $d=1,\ldots,n'$ and $\mathfrak{T}_d=\big\{i\in\{1,\ldots,n\}:\theta_i=\widehat{\theta}_d \big\}$,  the corresponding shape parameters  $\widehat{\gamma}_{d,l}=\sum_{i\in \mathfrak{T}_d}\gamma_{i,l},\ l=1,\ldots,m$. We would then proceed with the proof, as it is outlined below.)
Further, without loss of generality, assume that the shape parameters are ordered as ${\gamma}_{d,1}\le {\gamma}_{d,2}\le\cdots\le{\gamma}_{d,m}$, $d=1,\ldots,n$.

With the above in mind and for any BU, $j\in\mathcal{N}$, we have that the Laplace transform of the RV $S_X^{(\ast)_j}$ is:
\begin{align*}
\phi_{S_X^{(\ast)_j}}(t)&=\sum_{k=1}^m p^{(*)_j}_{k}\, \big(1+{\theta}_j\, t\big)^{-(1+{\gamma}_{j,k})}\prod_{d=1,d\neq j}^{n} \left(1+{\theta}_d\, t\right)^{-{\gamma}_{d,k}},\quad {\rm Re}(t)>0.
\end{align*}
Furthermore, as $\XX\in \mathfrak{W}$, we have that Theorem \ref{thm:main} implies, for $1\le i\neq j\le n$ and all ${\rm Re}(t)>0$,
\begin{align*}
\phi_{S_X^{(\ast)_{i}}}(t)=\phi_{S_X^{(\ast)_{j}}}(t).
\end{align*}
However, this is impossible, which is easily seen by comparing, e.g., the $m$-th terms of the Laplace transforms $\phi_{S_X^{(\ast)_{i}}}$ and $\phi_{S_X^{(\ast)_{j}}}$. Hence, we have arrived at a contradiction and the proposition is proved.
\end{proof}

\section{The case of independent losses}
\label{sec-ind}

In this section, we explore the loss portfolios $\XX\in\mathcal{X}^n$ that have independent constituents. Admittedly, the assumption of independence simplifies the problem postulated in Question \ref{question} considerably, yet nor it means that the RV $\RR=(R_1,\ldots,R_n)$ has independent coordinates, neither that the RVs $\RR$ and $S_X$ are independent, thus warranting a separate discussion.

\begin{theorem}
\label{prop:independent-DSI}
Assume that $\XX=(X_1,\ldots,X_n)\in\mathcal{X}^n$ is a portfolio of independent losses, then we have the equality $r_{q,i}=\tilde{r}_{q,i}(=\e[X_i]\,/\,\e[S])$ for all $q\in[0,\ 1)$ and $i\in\mathcal{N}$, if and only if the equality
\[
 \phi_{X_i}(t)= \left(\phi_{X_j}(t)\right)^{\e[X_i]/\e[X_j]}
\]
holds for all $i\neq j \in\mathcal{N}$ and ${\rm Re}(t)>0$.
\end{theorem}
\begin{proof}
By Theorem \ref{thm:main} and assuming that the RVs $X_1,\ldots,X_n$ are mutually independent, we have $\XX\in\mathfrak{W}$ if and only if the Laplace transforms of the RVs $X_i+X_j^\ast$ and $X_i^\ast +X_j$ agree for all $i\neq j\in\mathcal{N}$. That is, we must have
\begin{align*}
\frac{1}{\e[X_j]}\frac{\frac{d}{dt}\phi_{X_j}(t)}{\phi_{X_j}(t)}=\frac{1}{\e[X_i]}\frac{\frac{d}{dt}\phi_{X_i}(t)}{\phi_{X_i}(t)},
\end{align*}
for all $i\neq j\in\mathcal{N}$ and ${\rm Re}(t)>0$. This, in turn, is equivalent to
\begin{align}\label{thm-ind-eq-1}
\frac{\frac{d}{dt}\log \phi_{X_j}(t)}{\frac{d}{dt}\log \phi_{X_i}(t)}=\frac{\e[X_j]}{\e[X_i]},
\end{align}
implying, for all ${\rm Re}(t)>0$,
\begin{align}\label{thm-ind-eq-2}
\frac{\log \phi_{X_j}(t)}{\log \phi_{X_i}(t)}=\frac{\e[X_j]}{\e[X_i]}
\end{align}
by Pinelis' L'hopital-type rules. The fact that Equation \eqref{thm-ind-eq-2} leads to Equation \eqref{thm-ind-eq-1} is easy to check by routine differentiation in the latter equation. This completes the proof of the theorem.
\end{proof}

\subsection{Examples and further elaborations}
\label{ind-sex-subsec}

Examples of the loss portfolios $\XX\in\mathcal{X}^n$ that have independent constituents and also belong to the set $\mathfrak{W}$ are really numerous. For instance, consider losses $X_i$, $i\in\mathcal{N}$, that have infinitely divisible distributions and such that the condition in Theorem \ref{prop:independent-DSI} holds, then we have $\tilde{r}_{q,i}=r_{q,i}=\e[X_i]\,/\,\e[S_X]$ for any $q\in[0,\ 1)$. The next example enumerates some of the distributions of relevance, which play important roles in actuarial science and quantitative risk management.
\begin{example}
\label{ex:independent}
Assume that the portfolio of losses $\XX\in\mathcal{X}^n$ has independent constituents $X_1,\ldots,X_n$, then it belongs to the set $\mathfrak{W}$ given that these constituents have the following probability distributions:
\begin{itemize}
    \item $X_i\sim {\rm Negative}$-${\rm Binomial}(\beta_i,p)$, $\beta_i\in \Re_+$, $p\in(0,1)$, with mean $\e[X_i]=\beta_i\,(1-p)\,/\,p$ and Laplace transform  \begin{align*}
       \phi_{X_i}(t)= \Big(\frac{p}{1-(1-p)e^{-t}} \Big)^{\beta_i},\quad {\rm Re}(t)>0.
   \end{align*}
    \item $X_i\sim {\rm Gamma}(\gamma_i, \beta)$, $\gamma_i\in\Re_+$, $\beta\in \Re_+$, with mean $\e[X_i]=\gamma_i\,\beta$ and Laplace transform  \begin{align*}
        \phi_{X_i}(t)= (1+\beta\, t)^{-\gamma_i},\quad {\rm Re}(t)>0.
    \end{align*}
        \item $X_i\sim {\rm Inverse}$-${\rm Gaussian}(\mu_i, \mu_i^2)$, $\gamma_i\in\Re_+$, with mean $\e[X_i]=\mu_i$ and Laplace transform  \begin{align*}
        \phi_{X_i}(t)= \exp\Big\{\mu_i (1-\sqrt{1-2t}) \Big\},\quad {\rm Re}(t)>0.
    \end{align*}
\end{itemize}
\end{example}

\section{Further generalizations and afterthoughts}
\label{sec-afterthoughts}

We mentioned in Section \ref{Sec-2} that the CTE risk measure is a member of the class of weighted risk measures and that it induces the RC allocation $r_{q,i}$, $q\in[0,\ 1)$, $i\in\mathcal{N}$. In fact, a more encompassing class of risk measures - and hence a generalization of the CTE risk measure - can be defined as follows  \citep[][]{Furman2008a}. Let $v,w:[0,\ \infty)\rightarrow [0,\ \infty)$ be two (non-decreasing) functions, then the \textit{generalized weighted} risk measure is the map $H_{v,w}:\mathcal{X}\rightarrow [0,\ \infty)$, which, when well-defined and finite, is given by
\begin{align}
\label{eq:gen-w-rm}
H_{v,w}(X)=\frac{\e\big[v(X)\, w(X)\big]}{\e\big[w(X)\big]},\quad X\in\mathcal{X}.
\end{align}
For $w(x)=\mathbbm{1}\{x\geq \VaR_q(X)\}$ and $v(x)=x$, where $x\in[0,\ \infty)$ and $q\in [0,\ 1)$, we have that the generalized weighted risk measure reduces to the CTE risk measure.

Further, for $k\in\mathbb{N}$, set $v(x)=x^k,\ x\in[0,\ \infty)$ and keep the \textit{weight} function $x\mapsto w(x)$  equal the indicator function as before in order to emphasize the tail loss scenarios, then generalized weighted risk measure \eqref{eq:gen-w-rm} yields the $k$\textit{-th order} CTE risk measure. Furthermore, extending the notation in Section \ref{Sec-2}, let, for $q\in [0,\ 1),\ k\in\mathbb{N}$ and $i\in\mathcal{N}$ 
\[
\tilde{r}^k_{q,i}=\e[R_i^k|\ S>s_q]\quad \mbox{and}\quad {r}^k_{q,i}=\e[X_i^k|\ S>s_q]\,/\,\e[S^k|\ S>s_q].
\]
In general, the proportional $k$-th order CTE-based allocation ${r}^k_{q,i}$ is not fully-additive. Nevertheless, it is a meaningful quantity in quantitative risk management \citep[e.g.,][for elaborations and applications]{Furman2006,Kim2010,Landsmanetal2017}.

It is not difficult to see that, for a fixed $k\in\mathbb{N}$, the equality $\tilde{r}^k_{q,i}={r}^k_{q,i}$ holds for all $q\in [0,\ 1)$ and $i\in\mathcal{N}$, if an only if we have $\Cov(R_i^k, S^k\,|\, S>s_q)\equiv 0$. Therefore, it is natural to reformulate Question \ref{question} as follows.

\begin{question}
\label{question2}
For loss RVs $X_i\in L^k$, can we characterize those loss portfolios $\XX=(X_1,\ldots,X_n)\in\mathcal{X}^n$, for which the equality $\tilde{r}^k_{q,i}={r}^k_{q,i}$ holds for all $q\in[0,\ 1),\ i\in\mathcal{N}$, and a fixed $k\in\mathbb{N}$?
\end{question}

Question \ref{question2} seeks to characterize the RVs that belong to the set
\begin{align}
\label{eq:DSI-k}
  \mathfrak{W}_k=\big\{\XX=(X_1,\ldots,X_n)\in\mathcal{X}^n\ :\, r^k_{q,i}= \widetilde{r}^k_{q,i}\ \mbox{ for all $q\in [0,1)$ and $i\in\mathcal{N}$}\big\},
\end{align}
which we do next. To start off, note that if the equality $r^k_{q,i}= \widetilde{r}^k_{q,i}$ holds for all $q\in[0,\ 1)$, then setting $q=0$, implies that for all loss portfolios in the set $\mathfrak{W}_k$, we must have  $r^k_{q,i}= \widetilde{r}^k_{q,i}=\e[X_i^k]\,/\,\e[S_X^k]$.
\begin{theorem}
\label{thm:SB-character}
If the portfolio of losses $\XX=(X_1,\ldots,X_n)\in\mathcal{X}^n$ with $X_i\in L^k$ belongs to the set $\mathfrak{W}_k$ and so the equality $r^k_{q,i}=\tilde{r}^k_{q,i}(=\e[X_i^k]\,/\,\e[S_X^k])$ holds for all $q\in[0,\ 1),\ i\in\mathcal{N}$, and a fixed $k\in\mathbb{N}$, then $S_X^{(k)_i}=_d S_X^{(k)_j},\ 1\leq i\neq j\leq n$. The opposite direction does not hold.
\end{theorem}
\begin{proof}
The proof follows the same argumentation as in Theorem \ref{thm:main} with the quantities
$S_X,\ R_i$ and $G_i(s)=\mathbb{E}[R_i|S_X=s]-\mathbb{E}[R_i]$ replaced with the quantities $S_X^k,\ R_i^k$ and $G^k_i(s)=\mathbb{E}[R_i^k|S_X=s]-\mathbb{E}[R_i^k],\ s\in[0,\ \infty)$.

To see that the distributional equality, $S_X^{(k)_i}=_d S_X^{(k)_j},\ 1\leq i\neq j\leq n$, does not imply $\XX\in\mathfrak{W}_k$, consider the RV $(X_1,X_2)\in\mathcal{X}^2$ that has independent and identically distributed constituents; $X_i\sim {\rm Uni}[0,\, 1],\ i=1,2$. Clearly, we have $S_X^{(k)_1}=_d S_X^{(k)_2},\ k\in\mathbb{N}$. Nevertheless, with some algebra we obtain, for $i\in\{1,\, 2\}$,
\begin{align*}
    \e[X_i^k|S=s]=\frac{s^k}{1+k}\mathbbm{1}_{\{0\leq s<1\}}+\frac{1-(s-1)^{k+1}}{(1+k)(2-s)}\mathbbm{1}_{\{1\leq s\leq 2\}}, \quad s\in\mathbb{R}_+,
\end{align*}
which implies $\e[R_i^k|\ S_X=s]\neq \rm{const}$ for $k\neq1$ and hence $\tilde{r}^k_{q,i}\neq r^k_{q,i}$. This completes the proof of the theorem.
\end{proof}

According to Theorem \ref{thm:SB-character}, if the equalities $r_{q,i}=\tilde{r}_{q,i}$ and $r^2_{q,i}=\tilde{r}^2_{q,i}$ hold for all $q\in[0,\ 1)$ and $i\in\mathcal{N}$, then we must have $\e[R_i|\ S_X=s]\equiv \rm{const}$ and $\e[R_i^2|\ S_X=s]\equiv \rm{const}$, respectively. Therefore, the fact $\XX\in\mathfrak{W}_1$ does not imply $\XX\in\mathfrak{W}_2$ (also due to the counter example in the proof of Theorem \ref{thm:SB-character}). Next example demonstrates that this statement, when formulated in the opposite direction, does not hold either.
\begin{example}
Consider again the MG distribution as per Example \ref{ex:MMG-same-scale}, i.e., let $\XX\sim {\rm MG}_n(\boldsymbol{\gamma},\boldsymbol{\theta},\pp)$, and set $\theta_i\equiv\theta\in\mathbb{R}_+$ and
\begin{align*}
\frac{\gamma_{i,1}\, (\gamma_{i,1}+1)}{\left(\sum_{j=1}^n\gamma_{j,1}\right)\left(\sum_{j=1}^n\gamma_{j,1}+1\right)}=\cdots=\frac{\gamma_{i,m}\, (\gamma_{i,m}+1)}{\left(\sum_{j=1}^n\gamma_{j,m}\right)\left(\sum_{j=1}^n\gamma_{j,m}+1\right)}
\end{align*}
for all $i\in\mathcal{N}$. Then it is not difficult to check directly that $r_{q,i}^2=\tilde{r}_{q,i}^2,\ i\in\mathcal{N}$, and therefore we have $\mathbf{X}\in\mathfrak{W}_2$. However, the choice of parameters above does not guarantee Equation \eqref{p-new-eq},
and consequently by Theorem \ref{thm:main}, we do not necessarily have
$\XX\in \mathfrak{W}_1$.
\end{example}
We conclude this section by outlining a situation in which the fact, $\XX\in\mathfrak{W}_1$, does imply the fact, $\XX\in\mathfrak{W}_2$, and vice versa; curiously, this connects Questions \ref{question} and \ref{question2} to the celebrated Lukacs theorem \citep[][]{Lukacs1955}. For this, recall that the fact that the loss portfolio $\XX=(X_1,\ldots,X_n)\in\mathcal{X}^n$ belongs to the set $\mathfrak{W}_k,\ k\in\mathbb{N}$, or in other words, that the RVs $R_i$, $i=1,\ldots,n$, and $S_X$ are uncorrelated conditionally on $S_X>s$ for all $s\in[0,\ \infty)$, does not in general imply the fact that the loss RVs $\RR$ and $S_X$ are independent. This statement holds true even if the constituents of the loss portfolio $\XX\in\mathcal{X}^n$ are independent. The following assertion delineates the cases, in which the RVs $\RR$ and $S_X$ are independent in the context of Questions \ref{question} and \ref{question2}.

\begin{corollary}
Assume that the loss RVs $X_1,\ldots,X_n\in\mathcal{X}$ are independent. The loss portfolio $\XX=(X_1,\ldots,X_n)\in\mathcal{X}^n$ belongs to the sets $\mathfrak{W}_1$ and $\mathfrak{W}_2$ if and only if the loss RV $X_i\in\mathcal{X}$ is distributed gamma with the shape and scale parameters $\gamma_i>0$ and $\theta>0$, respectively, $i\in\mathcal{N}$. In this case, the RVs $\RR$ and $S_X$ are independent.
\end{corollary}
\begin{proof}
\label{appendixproof2}
In order to prove the `if' direction, we note that by Lukacs' theorem, the assumption $X_i\sim Ga(\gamma_i,\theta)$ implies $R_i\perp\!\!\!\perp S_X$ for all  $i\in\mathcal{N}$, which in turn implies 
$\XX \in \mathfrak{W}_1\cap\mathfrak{W}_2$.

Further, let us prove the `only if' direction. For this, fix $i\in\mathcal{N}$ and note that by Theorems \ref{prop:independent-DSI} and \ref{thm:SB-character} - with the assumption of independence made in the latter case - we arrive at the following two equations, where $\phi_{X_i}(t)$ and $\phi_{X_j}(t)$ denote, respectively, the Laplace transforms of the loss RVs $X_i$ and $X_j$, and $a:={\mathbb{E}[X_j]\,/\, \mathbb{E}[X_i]}$ and $b:={\mathbb{E}[X_i^2]\,/\,\mathbb{E}[X_j^2]},\ 1\leq i\leq j\leq n$, for the simplicity of exposition
\begin{equation}
\label{eq:B1}
\phi_{X_j}(t)=\left(\phi_{X_i}(t)\right)^{a},\quad {\rm Re}(t)>0
\end{equation}
and
\begin{equation}
\label{eq:B2}
\phi_{X_j}(t)\,\frac{d^2}{dt^2}\phi_{X_i}(t)=b\,\phi_{X_i}(t)\,\frac{d^2}{dt^2}\phi_{X_j}(t),\quad {\rm Re}(t)>0.
\end{equation}
Rewriting the equations above in terms of the Laplace transform $\phi_{X_i}$ only, we obtain, for ${\rm Re}(t)>0$,
\begin{align*}
\left(\phi_{X_i}(t)\right)^{a}\,\frac{d^2}{dt^2}\phi_{X_i}(t)
=b\, \phi_{X_i}(t)\left[a(a-1)\left(\phi_{X_i}(t)\right)^{a-2}\left(
\frac{d}{dt}\phi_{X_i}(t)\right)^2+a\,\left(\phi_{X_i}(t)\right)^{a-1}
\frac{d^2}{dt^2}\phi_{X_i}(t)\right],
\end{align*}
which, with some algebra and the notation $c=a\, b\,(a-1)\,(1-a\, b)^{-1}$, simplifies to
\begin{align*}
\frac{\frac{d^2}{dt^2}\phi_{X_i}(t)}{\frac{d}{dt}\phi_{X_i}(t)}=c\, \frac{\frac{d}{dt}\phi_{X_i}(t)}{\phi_{X_i}(t)},\quad {\rm Re}(t)>0.
\end{align*}
After some more algebra, we arrive at the following first-order non-linear ODE
\begin{align*}
\frac{d}{dt}\phi_{X_i}(t)=-\e[X_i]\, \left(\phi_{X_i}(t)\right)^c,\quad {\rm Re}(t)>0,
\end{align*}
with the solution
\begin{align*}
\phi_{X_i}(t)=\Big[1+(c-1)\,\e[X_i]\,t \Big]^{-1/(c-1)}.
\end{align*}
Finally, substituting the expressions for the first and second moments of the gamma distribution in the constant $c$ and hence noticing that $c-1=1/\gamma_i$, we arrive at
\begin{align*}
\phi_{X_i}(t)=\left(1+\theta\, t \right)^{-\gamma_i},\quad {\rm Re}(t)>0,
\end{align*}
which is the Laplace transform of the RV distributed gamma with the shape and scale parameters $\gamma_i>0$ and $\theta>0$, respectively. Hence, $X_i\backsim Ga(\gamma_i,\sigma),\ i\in\mathcal{N}$. Also, for $\gamma_{\bullet}=\gamma_1+\cdots+\gamma_n$ as before, we have $S_X\sim Ga(\gamma_{\bullet},\sigma)$, and by Lukacs' theorem, the RVs $R_i=X_i\,/\,S_X$ and $S_X$ are independent. This completes the proof of the `only if' direction as well as of the corollary as a whole.
\end{proof}
\section{Conclusions}
\label{Sec-conclusions}	

The striking majority of the risk capital allocation rules that exist nowadays are induced by risk measures, with the CTE risk measure being arguably the most popular. While the CTE risk measure - and the risk capital allocation based on it - are sound mathematical objects that have been studied extensively in a great variety of contexts and even embedded in some regulatory accords, they are not explicitly linked to the risk preferences of an insurer. Yet, these risk preferences are of fundamental importance and should presumably drive - or at least impact - the allocation exercise.

In this paper, we have demonstrated that there are model settings, in which the risk capital allocation induced by the CTE risk measure does reflect insurer's risk preferences by accounting for the decreasing marginal effect of the increase in the aggregate loss. This implies that there exist loss portfolios for which the utilization of the CTE risk measure as a basis for allocating risk capital may be justified from both the regulatory and profit maximizing perspectives. Moreover, we have developed exhaustive description of the just-mentioned special loss portfolios and elucidated the findings with ample specific examples.

Our message in this paper is two-fold. On the one hand-side, the feasibility of achieving the apparently desirable agreement between the prudence-  and profitability- driven considerations for some loss portfolios gives hope that the CTE-based risk capital allocation may indeed be \textit{appropriate} for profitability analysis and price determination. On the other hand-side, the sparsity of the collection of the loss portfolios for which the just-mentioned agreement is achieved implies that more research is required in order to explicitly connect the risk capital allocation exercise to its objectives as well as to the economic traits of the insurer.

\bibliographystyle{apalike}
\bibliography{SC_Allocation.bib}

\end{document}